\documentclass[12pt,a4paper]{article}

\usepackage{amsmath,amsfonts,amssymb,amsthm,mathrsfs}

\usepackage[a4paper,top=30mm,bottom=30mm,left=25mm,right=25mm]{geometry}

\usepackage[blocks]{authblk}
\usepackage[utf8]{inputenc}
\usepackage[british]{datetime2}
\DTMlangsetup[en-GB]{ord=omit}   

\date{\DTMtoday}                 

\usepackage[pagewise]{lineno} 
\newtheorem{thm}{Theorem}[section]
\newtheorem{ex}[thm]{Example}
\newtheorem{coro}[thm]{Corollary}

\newtheorem{defn}[thm]{Definition}
\newtheorem{rem}[thm]{Remark}
\newtheorem{prop}[thm]{Proposition}

\newcommand{\Schrodinger}{{Schr\"odinger }}

\newcommand{\cadlag}{{c\`adl\`ag }}
 \def\1{\mathds{1}}
\def\R{\mathbb{R}}

\def\M{\mathbb{M}}
\def\N{\mathbb{N}}

\def\T{\mathbb{T}}

\def\dd{\mathrm{d}}

\def\law{\overset{\textnormal{law}}{=}}

\def\F{\mathscr{F}}
\def\G{\mathscr{G}}

\title{Martingale Projections and Quantum Decoherence}

\author[1]{Lane P.~Hughston}
\author[2,3]{Levent A.~Meng\"ut\"urk}

\affil[1]{School of Computing, Goldsmiths University of London\par
New Cross, London SE14 6NW, UK}

\affil[2]{Department of Mathematics, University College London\par
25 Gordon Street, London 
WC1H 0AY, UK}

\affil[3]{Artificial Intelligence and Mathematics Research Lab\par
James Carter Road, Mildenhall, Bury St Edmunds IP28 7DE, UK}

\date{\DTMtoday}

\begin{document}
\maketitle

\begin{abstract}
\noindent
We introduce so-called super/sub-martingale projections as a family of endomorphisms defined on unions of Polish spaces. Such projections allow us to identify martingales as collections of transformations that relate path-valued random variables to each other under conditional expectations. In this sense, super/sub-martingale projections are random functionals that (i) are boundedness preserving and (ii) satisfy a conditional expectation criterion similar to that of the classical martingale theory. As an application to the theory of open quantum systems, we prove (a) that any system-environment interaction that manifests a supermartingale projection on the density matrix gives rise to decoherence, and (b) that any system-environment interaction that manifests a submartingale projection gives rise an increase in Shannon-Wiener information. It follows (c) that \emph{martingale} projections in an open quantum system give rise both to quantum decoherence and to information gain.
\end{abstract}



\section{Introduction}
Most physical, biological and social systems observed in nature can be understood in terms of (a) an environment, (b) a set of admissible actions taking place in that environment, and (c) the set of all possible states formed as a manifestation of these actions. In essence, it is viable to consider any such dynamical system as some form of open system that hosts environmental interactions and transformations, which, in turn, evolve and shape the system as a consequence of transmissions. In light of this, we focus on quantum systems as open systems consisting of states and actions on those states. 
More specifically, we shall study \emph{decoherence}, a well-known physical phenomenon that involves the reduction of the off-diagonal coordinates of a quantum density matrix in the basis of an observable as a result of internal and external interactions \cite{Zeh 1970, Joos Zeh 1985,   Joos et al 2003, Schlosshauer 2004, Zurek 2003 RMP,  Schlosshauer 2005, Schlosshauser 2010,  Zurek 2022}. Decoherence can be regarded as a fundamental aspect of the quantum theory of open systems \cite{Davies 1976, Gisin Percival 1992,  Breuer Petruccione 2004, Rivas Huelga 2012}. We highlight that quantum decoherence has found applications in quantum information processing \cite{Shabani Lidar 2009, Breuer Amato Vacchini 2018, Zurek 2025} and quantum computing \cite{ Preskill 1998, Zurek 2003 PRL, Shor 1995} and has connections to quantum thermodynamics \cite{Gemmer et al 2009, Popescu et al 2006, Esposito et al 2009}. Direct empirical insights into the nature of decoherence have been gained by the use of controlled systems such as trapped ions, superconducting qubits, and cavity QED \cite{Monroe et al 1996, Brune et al 1996, Myatt et al 2000}. Although quantum decoherence has generally been interpreted, at least informally, as information loss due to the flow of information from a system into its environment \cite{Zurek 1991, Zurek 2009}, it has recently been shown in \cite{Brody Hughston 2025} that the opposite holds: decoherence implies information \emph{gain}. 

The work of reference \cite{Brody Hughston 2025} motivates the present investigation, in which we take the view that quantum decoherence is a fundamentally important concept that provides deeper insights into the state evolution of a quantum system that interacts with its environment. The original contributions of our work here are two-fold: (i) the introduction of what we call super/sub-martingale projections as random endomorphisms of path-valued random variables on unions of Polish spaces in a broad mathematical setting, and (ii) the proof of a fundamental connection of super/sub-martingale projections with quantum decoherence and Shannon-Wiener information. The motivation behind generalizing classical super/sub-martingale processes to super/sub-martingale projections on path spaces is to model possibly segmented, discrete or event-triggered system-environment interactions rather than asking global conditions on the relevant stochastic processes. More specifically, super/sub-martingale projections -- and therefore martingale projections -- bring the flexibility to compartmentalize random system-environment interactions across time, instead of asking such features from the \emph{entire} evolution even if not existent on a continual basis, which will be important when investigating quantum decoherence.

One of the most highly developed mathematical foundations for the enrichment of  our understanding of quantum decoherence is quantum stochastic calculus, particularly in connection with its applications in describing unitary dilations of open quantum systems through quantum stochastic differential equations \cite{Hudson Parthasarathy 1982, Hudson Parthasarathy 1984, Parthasarathy Sinha 1986, Parthasarathy 1992, Meyer 1993, Fagnola 1994, Fagnola Wills 2003, Chang 2015}. In this active research field, a fundamental structure is quantum noise, which is an operator-valued stochastic process defined on a given interacting environment (e.g. bosonic field) represented by a Fock space. Consequentially, the Fock space produces a natural quantum It\^o formalism that enables quantum decoherence to be derived dynamically through quantum master equations. Our framework differs from this approach, since we do not bring forward (quantum) stochastic calculus, but instead work with path-valued random variables and their transformations. If the considered paths over a chosen Polish space turn out to be solutions to stochastic differential equations, then path transformations become implicitly linked to such dynamical representations (we shall provide an example later in the paper), but we do not need to surface them in our study of quantum decoherence. In fact, this flexibility is one of our core motivations to introduce super/sub-martingale projections as mentioned above. In addition, since Fock space is a Hilbert space completion of the direct sum of symmetric/antisymmetric tensor powers of a separable Hilbert space, it is a Polish space when equipped with its natural norm topology. Accordingly, operator-valued stochastic processes from quantum stochastic calculus are hosted as path-valued random variables in our broad mathematical setting, which we find stimulating as a dedicated study in future. Nonetheless, we shall note here that by choosing the Polish space in our framework as a Fock space, one can expect to generalise quantum martingales through martingale projections over the given Fock space.

Although we believe super/sub-martingale projections could merit a dedicated study in their own right, we adopt a parsimonious presentation of these functionals towards the study of quantum decoherence. Accordingly, we prove (a) that the existence of a supermartingale projection on the diagonal magnitudes of a density matrix implies quantum decoherence, and (b) that the existence of a submartingale projection on the diagonal magnitudes of a density matrix implies increasing Shannon-Wiener information, under conditional expectations. Hence, (c) if a system-environment interaction manifests a martingale projection on the density matrix of an open quantum system, one can expect both decoherence and information gain within that system. 

The mathematical framework of the paper is structured as follows. In Sections 2 and 3 we introduce super/sub-martingale projections as a family of endomorphisms on Polish spaces of random variables. Thereafter, we continue in Section 4 within a quantum framework, where we study decoherence and Shannon-Wiener information in terms of super/sub-martingale projections. Our conclusions are summarized in Section 5 and in the Appendix we provide some examples of the structures under consideration, as well as a table of key mathematical symbols used in the paper to aid the reader.

\section{Path Transformations}
We work with a probability space $(\Omega,\F,\nu)$ and a nonempty compact set $\T\subset\R_+$. In essence, we ask $\T$ to be a totally-ordered set and we choose its superset $\R_+ = [0,\infty)$ canonically. Accordingly, any element of $\T$ can be interpreted as a point in \emph{time}; hence, $\T$ and any of its subsets can be taken to represent time-frames. Hereafter, $\triangleq$ stands for ``equals to by definition". We define a sequence of nonempty subsets
\begin{align}
\T_{[r,t]} \,\triangleq \, [r,t] \,\bigcap \,\T \hspace{0.1in} \text{for any $r\leq t\in\T$} \,\,\,\text{\&}\,\,\, \T_{[r,t]} \neq \emptyset,
\end{align}
where $[r,t] = \{s \,:\, r \leq s \leq t\}$ is a closed interval. Note that 
\begin{align}
\emptyset \subset \T_{[r,s]} \subseteq \T_{[r,t]} \subseteq \T, \notag 
\end{align}
for any $r\leq s\leq t \in \T$. Also, for any $r\leq t \in \T$, we have 
\begin{align}
\inf(\T_{[r,t]}) = r \,\,\, \text{and} \,\,\, \sup(\T_{[r,t]}) = t. \notag
\end{align}
\begin{rem} {\em
If $\rho = \inf(\T)$ and $\tau = \sup(\T)$, then $\T_{[\rho,\tau]} = \T$. Also, note that $\T_{[t,t]} = \{t\}$ for any $t\in\T$.}
\end{rem}
Let $\mathbb M$ be a Polish space -- that is to say, a topological space that is separable and completely metrizable (e.g. $\mathbb{R}^n$ for any finite $n\geq1$). Let $(\mathbb M,\mathscr{B}(\mathbb M))$ be its Borel space. Since $\T_{[r,t]}$ is a Polish space for any $r\leq t\in\T$, the product $\mathbb M\times\T_{[r,t]}$ is also a Polish space -- hence, $(\mathbb M\times\T_{[r,t]},\mathscr{B}(\mathbb M\times\T_{[r,t]})$ is a Borel space. For what follows, we are interested in $(\mathbb M\times\T_{[r,t]})$-valued random variables for any $r\leq t\in\T$. 
Note that each $(\mathbb M\times\T_{[r,t]})$ can be considered a space of \emph{paths} for any $r\leq t\in\T$, where any $\mathbb M$-valued random variable collected over a time-frame $\T_{[r,t]}$ is itself a random variable that can be viewed as a trajectory in $(\mathbb M\times\T_{[r,t]})$.
Finally, we define the double-union space
\begin{align}
\boldsymbol{\Lambda}(\mathbb M \,, \T) = \bigcup_{t\in\T} \,\, \bigcup_{r\in \T \, \bigcap \, r \leq t} \left(\mathbb M\times\T_{[r,t]}\right),
\end{align}
which need \emph{not} be a Polish space -- we recall that unions of Polish spaces are not neccesarily Polish. Note that $\boldsymbol{\Lambda}(\mathbb M \,, \T)$ is itself a path space, where any $\mathbb M$-valued random variable collected over any $\T^*\subseteq\T$ is a path in $\boldsymbol{\Lambda}(\mathbb M \,, \T)$. 

Henceforth, let $X_t$ be an $\mathbb M$-valued random variable for any fixed $t\in\T$. Accordingly, we write the collection over any $\T^*\subseteq\T$:
\begin{align}
\boldsymbol{X}_{\T^*} \triangleq \{X_t \, : \, \forall t\in\T^*\} \,\, \text{for $\T^*\subseteq\T$} \,\,\,\text{\&}\,\,\, \T^* \neq \emptyset, \label{pathnotation}
\end{align}
where $\T^*\subseteq\T$ is totally ordered. Note that $\boldsymbol{X}_{\T^*}$ in (\ref{pathnotation}) is an element of $(\mathbb M\times\T^*) \subseteq \boldsymbol{\Lambda}(\mathbb M \,, \T)$ for any nonempty $\T^*\subseteq\T$. In other words, $\boldsymbol{X}_{\T^*}$ is an $(\mathbb M\times\T^*)$-valued random variable, and hence, a $\boldsymbol{\Lambda}(\mathbb M \,, \T)$-valued random variable.
\begin{rem}
\label{remone}{\em
If $\T^* = \{t\}$ for fixed $t\in\T$, then $\boldsymbol{X}_{\T^*} = X_t \in \boldsymbol{\Lambda}(\mathbb M \,, \T)$ is a \emph{singleton path}, since we can view it as an $(\mathbb M\times\{t\})$-valued random variable.}
\end{rem}
We are now in position to introduce a family of endomorphisms on $\boldsymbol{\Lambda}(\mathbb M \,, \T)$. 
\begin{defn}
\label{transformationsgeneral} {\em
For any $p\leq r \in \T$ and any $s\leq t \in \T$, let 
\begin{align}
\boldsymbol{\Gamma}^{(p,r)\mapsto(s,t)} \,: \, \boldsymbol{\Lambda}(\mathbb M \,, \T) \rightarrow \boldsymbol{\Lambda}(\mathbb M \,, \T)  
\end{align}
be an endomorphism such that
\begin{align}
\boldsymbol{Y}_{\T_{[s,t]}} = \boldsymbol{\Gamma}^{(p,r)\mapsto(s,t)}\left( \boldsymbol{X}_{\T_{[p,r]}}\right), \label{mapgeneralformula}
\end{align}
with $\boldsymbol{X}_{\T_{[p,r]}} \in \boldsymbol{\Lambda}(\mathbb M \,, \T)$ and $\boldsymbol{Y}_{\T_{[s,t]}} \in \boldsymbol{\Lambda}(\mathbb M \,, \T)$. Then we say $\boldsymbol{\Gamma}^{(p,r)\mapsto(s,t)}$ is a \emph{path transformation} on $\boldsymbol{\Lambda}(\mathbb M \,, \T)$ from $\T_{[p,r]}$ to $\T_{[s,t]}$. }
\end{defn}
We refer the reader to the Appendix for various examples. We shall however provide a simple visualization of a path transformation for motivation, below.
\begin{ex}
\emph{\label{examplesimplebegin}
Set $\mathbb M=\R$ and $\T=[0,3\pi]$. Let $\boldsymbol{X}_{\T_{[0,2\pi]}} = \boldsymbol{X}_{[0,2\pi]}$ be given by $X_t = \sin(t) + W_t$ for every $t\in[0,2\pi]$, where $\{W_t\}_{t\in[0,2\pi]}$ is an $\R$-valued Gaussian process under $\nu$. Define a path transformation
\begin{equation}
    \boldsymbol{Y}_{[\pi,3\pi]} = \boldsymbol{\Gamma}^{(0,2\pi)\mapsto(\pi,3\pi)}\left(\boldsymbol{X}_{[0,2\pi]}\right) =
\begin{cases}
			X_t, & \text{if}\ t\in [\pi,2\pi) \\
      X_{2\pi} + \sin(t) + \epsilon\mathbf{1}(\tau \leq t), & \text{if}\ t\in[2\pi,3\pi],
\end{cases}
\end{equation}
where $\epsilon: \Omega \rightarrow \R $ and $\tau: \Omega \rightarrow \R_+$ are $\F$-measurable random variables, with $\mathbf{1}(.)$ the indicator function such that $\mathbf{1}(\tau \leq t)$ equals 1 if $\tau \leq t$ and 0 otherwise.}
\end{ex} 
We emphasize that path transformations in Definition \ref{transformationsgeneral} can generate random transformations -- this can also be seen in Example \ref{examplesimplebegin} through the appearance of random variables $\epsilon$ and $\tau$. More specifically, define a subalgebra $\G \subset \F$. Even if $\boldsymbol{X}_{\T_{[p,r]}} \in \boldsymbol{\Lambda}(\mathbb M \,, \T)$ is $\G$-measurable, the random variable $\boldsymbol{Y}_{\T_{[s,t]}}\in \boldsymbol{\Lambda}(\mathbb M \,, \T)$ in (\ref{mapgeneralformula}) need not be $\G$-measurable. In addition, the statement below highlights that compositions of path transformations are path transformations.
\begin{prop}
Let $\mathcal{I} = \{1,\ldots,K\}$ be an ordered set for an integer $K\geq 2$ and let $\boldsymbol{\Gamma}^{\left(\alpha^{(i)},\beta^{(i)}\right)\mapsto\left(\alpha^{(i+1)},\beta^{(i+1)}\right)}$ be a path transformation for every $i\in\mathcal{I}\setminus\{K\}$. Then, for any $\alpha^{(i)}\leq \beta^{(i)} \in \T$ for $i\in\mathcal{I}$, the composition
\begin{align}
\boldsymbol{\Gamma}^{\left(\alpha^{(1)},\beta^{(1)}\right)\mapsto\left(\alpha^{(K)},\beta^{(K)}\right)} = \boldsymbol{\Gamma}^{\left(\alpha^{(1)},\beta^{(1)}\right)\mapsto\left(\alpha^{(2)},\beta^{(2)}\right)} \,\, \circ \,\, \cdots \,\, \circ \,\, \boldsymbol{\Gamma}^{\left(\alpha^{(K-1)},\beta^{(K-1)}\right)\mapsto\left(\alpha^{(K)},\beta^{(K)}\right)} \label{compositionmapped}
\end{align} 
is a path transformation on $\boldsymbol{\Lambda}(\mathbb M \,, \T)$ from $\T_{[\alpha^{(1)},\beta^{(1)}]}$ to $\T_{[\alpha^{(K)},\beta^{(K)}]}$.
\end{prop}
\begin{proof}
Since $\boldsymbol{\Gamma}^{\left(\alpha^{(i)},\beta^{(i)}\right)\mapsto\left(\alpha^{(i+1)},\beta^{(i+1)}\right)}$ is a path transformation, it is an endomorphism on $\boldsymbol{\Lambda}(\mathbb M \,, \T)$ for every $i\in\mathcal{I}\setminus\{K\}$ from Definition \ref{transformationsgeneral}. Composition of any two endomorphisms is itself an endomorphism, and thus, $\boldsymbol{\Gamma}^{\left(\alpha^{(1)},\beta^{(1)}\right)\mapsto\left(\alpha^{(K)},\beta^{(K)}\right)}$ is an endomorphism on $\boldsymbol{\Lambda}(\mathbb M \,, \T)$. Since every $\alpha^{(i)}\leq \beta^{(i)} \in \T$ is an ordered pair in $\T$, the chain of compositions in (\ref{compositionmapped}) must be a path transformation from $\T_{[\alpha^{(1)},\beta^{(1)}]}$ to $\T_{[\alpha^{(K)},\beta^{(K)}]}$, using (\ref{mapgeneralformula}).
\end{proof}
In Definition \ref{transformationsgeneral}, if we set $s=t$ in (\ref{mapgeneralformula}), we can write
\begin{align}
\boldsymbol{Y}_{\T_{[t,t]}} = \boldsymbol{\Gamma}^{(p,r)\mapsto(t,t)}\left( \boldsymbol{X}_{\T_{[p,r]}}\right) = Y_t, \label{projectionformula}
\end{align}
using Remark \ref{remone}. Hence, $\boldsymbol{\Gamma}^{(p,r)\mapsto(t,t)}$ can be viewed as a \emph{projection} on $\boldsymbol{\Lambda}(\M \,, \T)$ for any $p\leq r \in \T$ and any $t \in \T$. In addition, if also $p=r$, then we can write
\begin{align}
Y_t = \boldsymbol{\Gamma}^{(r,r)\mapsto(t,t)}\left( X_r \right), \label{projectionformulatwo}
\end{align}
for any $r,t \in \T$. For such transformations, we can ease the notation by writing
\begin{align}
\boldsymbol{\Gamma}^{(p,r)\mapsto(t)} \,\triangleq \, \boldsymbol{\Gamma}^{(p,r)\mapsto(t,t)} \,\,\,\, \text{and} \,\,\,\, \boldsymbol{\Gamma}^{(r)\mapsto(s,t)} \,\triangleq \, \boldsymbol{\Gamma}^{(r,r)\mapsto(s,t)} \,\,\,\, \text{and} \,\,\,\, \boldsymbol{\Gamma}^{(r)\mapsto(t)} \,\triangleq \, \boldsymbol{\Gamma}^{(r,r)\mapsto(t,t)} 
\end{align} 
for $p,r,s,t \in \T$. We employ this simplification whenever we have a repeated index symbol. For what follows, $X \law Y$ means $X$ and $Y$ have the same probability distribution. 
\begin{rem}
\label{stochasticassociation} {\em
Let $\{X_t\}_{t\in\T}$ be an $\mathbb M$-valued stochastic process on $\T$ and $\rho=\inf(\T)$. For some fixed $s\in \T$, if
\begin{align}
\left. X_u \, \right|_{\boldsymbol{X}_{\T_{[\rho,s]}}} \law \boldsymbol{\Gamma}^{(\rho,s)\mapsto(u)}\left( \boldsymbol{X}_{\T_{[\rho,s]}}\right), \label{processgenerator}
\end{align}
for every $u \in \T_{[s,\tau]}$ with $\tau=\sup(\T)$, the family of path transformations 
\begin{align}
\mathcal{P}^{(\rho,s)} \triangleq \left\{\boldsymbol{\Gamma}^{(\rho,s)\mapsto(u)} \, : \, \forall u \in \T_{[s,\tau]} \right\} \label{familypathransforms}
\end{align}
is law-consistent with $\{X_t\}_{t\in\T}$, given its path $\boldsymbol{X}_{\T_{[\rho,s]}}$. Note that if $s=\rho$ in (\ref{processgenerator}), then $\mathcal{P}^{(\rho,\rho)}$, if it exists, can generate $\{X_t\}_{t\in\T}$ from the initial state $X_{\rho}$. }
\end{rem}

\section{Martingale Projections}
We shall now introduce a special class of path transformations relevant to our study. We start by defining the $\sigma$-algebras
\begin{align}
\G^{\boldsymbol{X}}_{\T_{[r,t]}} = \sigma\left( \{X_s \, : \, \forall s\in\T_{[r,t]}\} \right) = \sigma\left( \boldsymbol{X}_{\T_{[r,t]}} \right) \, \subseteq \, \G_{\T_{[r,t]}} , \notag
\end{align}
where $\G_{\T_{[r,t]}} \subset \F$ for any $r\leq t\in\T$. Therefore, $\boldsymbol{X}_{\T_{[r,t]}}$ is $\G_{\T_{[r,t]}}$-measurable. Also, we clearly have 
\begin{align}
\G^{\boldsymbol{X}}_{\T_{[r,s]}} \subseteq \G^{\boldsymbol{X}}_{\T_{[r,t]}} \,\,\,\, \text{and} \,\,\,\, \G_{\T_{[r,s]}} \subseteq \G_{\T_{[r,t]}} \hspace{0.1in} \text{$\forall r\leq s\leq t \in \T$.}
\end{align}
\begin{rem}
\label{existenceconditionallaw} {\em
Since $\M\times\T_{[r,t]}$ is a Polish space for any $r\leq t\in\T$,  there exists a regular conditional distribution for any $(\M\times\T_{[r,t]})$-valued random variable given any $\G \subset \F$. }
\end{rem}
Remark \ref{existenceconditionallaw} allows us to work with conditional expectations, which are key functionals for our purposes. For example, if a path transformation is bounded in conditional expectation on a bounded subset of $\boldsymbol{\Lambda}(\M \,, \T)$, we call this property \emph{boundedness preserving}. 
\begin{defn} {\em
A path transformation $\boldsymbol{\Gamma}^{(p,r)\mapsto(s,t)}$ is said to be boundedness preserving with respect to $\G_{\T_{[p,r]}}$ under $\nu$-measure, if for a bounded $\boldsymbol{S}(\M \,,\T) \subseteq \boldsymbol{\Lambda}(\M \,, \T)$ there exists a bounded $\boldsymbol{S}^*(\M \,, \T) \subseteq \boldsymbol{\Lambda}(\M \,, \T)$ such that
\begin{align}
\boldsymbol{X}_{\T_{[p,r]}} \in \boldsymbol{S}(\M \,, \T) \,\, \Rightarrow \,\, \mathbb{E}^{\nu}\left[\boldsymbol{\Gamma}^{(p,r)\mapsto(s,t)}\left(  \boldsymbol{X}_{\T_{[p,r]}} \right)  \,\, \left| \,\, \G_{\T_{[p,r]}} \right.\right] \in \boldsymbol{S}^*(\M \,, \T), \label{boundednesspreservation}
\end{align}
for every $\boldsymbol{X}_{\T_{[p,r]}} \in \boldsymbol{S}(\M \,, \T)$. }
\end{defn} 
If $\boldsymbol{\Gamma}^{(p,r)\mapsto(s,t)}$ is boundedness preserving with respect to $\G_{\T_{[p,r]}}$ under $\nu$-measure, then we write 
\begin{align}
\boldsymbol{\Gamma}^{(p,r)\mapsto(s,t)} \in \mathbb{B}_{\nu,\G}(\boldsymbol{\Lambda}(\M \,, \T)). \notag 
\end{align}
We highlight that (\ref{boundednesspreservation}) is a technical condition to ensure that super/sub-martingale inequalities that we provide below, written in terms of conditional expectations of path transformations on $\boldsymbol{\Lambda}(\M \,, \T)$, are well-defined. Essentially, bounded paths do not ``explode" when subject to a transformation that belong to $\mathbb{B}_{\nu,\G}(\boldsymbol{\Lambda}(\M \,, \T))$. Hereafter, we work with an $\M$ that can be equipped with a partial order we denote as $\preceq$. Recall that a total order $\leq$ is a partial order.
\begin{defn}
\label{martingaletransformationdefinitionmain} {\em
For some $p\leq r \leq t \in \T$, let $\boldsymbol{\Gamma}^{(p,r)\mapsto(t)}$ be a projection path transformation such that $\boldsymbol{\Gamma}^{(p,r)\mapsto(t)} \in \mathbb{B}_{\nu,\G}(\boldsymbol{\Lambda}(\M \,, \T))$. Then $\boldsymbol{\Gamma}^{(p,r)\mapsto(t)}$ is a \emph{supermartingale projection} of $\boldsymbol{X}_{\T_{[p,r]}} \in \boldsymbol{\Lambda}(\M \,, \T)$ with respect to $\G_{\T_{[p,r]}}$ under $\nu$-measure if
\begin{align}
\mathbb{E}^{\nu}\left[\boldsymbol{\Gamma}^{(p,r)\mapsto(t)}\left(  \boldsymbol{X}_{\T_{[p,r]}} \right)  \,\, \left| \,\, \G_{\T_{[p,r]}} \right.\right] \preceq X_{r} \label{supermartingaleperc},
\end{align}
and a \emph{submartingale projection} of $\boldsymbol{X}_{\T_{[p,r]}} \in \boldsymbol{\Lambda}(\M \,, \T)$ with respect to $\G_{\T_{[p,r]}}$ under $\nu$-measure if
\begin{align}
X_{r} \preceq \mathbb{E}^{\nu}\left[\boldsymbol{\Gamma}^{(p,r)\mapsto(t)}\left(  \boldsymbol{X}_{\T_{[p,r]}} \right)  \,\, \left| \,\, \G_{\T_{[p,r]}} \right.\right] \label{submartingaleperc}.
\end{align}
If $\boldsymbol{\Gamma}^{(p,r)\mapsto(t)}$ satisfies both {\rm(\ref{supermartingaleperc})} and {\rm (\ref{submartingaleperc})}, and hence,
\begin{align}
\label{martingaleprojdefinition}
\mathbb{E}^{\nu}\left[\boldsymbol{\Gamma}^{(p,r)\mapsto(t)}\left(  \boldsymbol{X}_{\T_{[p,r]}} \right)  \,\, \left| \,\, \G_{\T_{[p,r]}} \right.\right] = X_{r},
\end{align}
then it is a \emph{martingale projection} of $\boldsymbol{X}_{\T_{[p,r]}} \in \boldsymbol{\Lambda}(\mathbb M \,, \T)$ with respect to $\G_{\T_{[p,r]}}$ under $\nu$-measure.}
\end{defn}
\begin{rem}{\em
If $\M$ does not admit a partial order, we define martingale projections directly from (\ref{martingaleprojdefinition}). If $\M$ admits multiple partial orders, then (\ref{supermartingaleperc}) and (\ref{submartingaleperc}) should be understood per a fixed $\preceq$ amongst all possible partial orders.}
\end{rem}

With Definition \ref{martingaletransformationdefinitionmain} at hand, let us clarify the connection between martingales from stochastic analysis \cite{Williams 1991, Karatzas Shreve 1991, Rogers Williams 2000} and martingale projections. Let $\{X_t\}_{t\in\T}$ be an $\mathbb M$-valued process on $\T$ and let $|| . ||_{\mathbb M}$ be a norm in $\mathbb M$. If $\{X_t\}_{t\in\T}$ satisfies
\begin{align}
&\mathbb{E}^{\nu}\left[ || X_t ||_{\mathbb M} \right] < \infty, \,\,\, \text{$\forall t \in \T$} \label{classicmartingaleone} \\
&\mathbb{E}^{\nu}\left[ X_t \,\, \left| \,\, \G_{\T_{[0,s]}} \right.\right] = X_{s}, \,\,\, \text{$\forall \, s \leq t \in \T$}, \label{classicmartingaletwo}
\end{align}
then $\{X_t\}_{t\in\T}$ is a $(\nu,\G)$-martingale. Note that the property of being a martingale is a property of a \emph{stochastic process}, since the conditions in (\ref{classicmartingaleone})-(\ref{classicmartingaletwo}) must hold at \emph{every} element of $\T$. In essence, a martingale projection relaxes this requirement and asks whether there exists a projection from any path (which may be a singleton path) in $\boldsymbol{\Lambda}(\mathbb M \,, \T)$ onto a point (which \emph{is} a singleton path) in $\boldsymbol{\Lambda}(\mathbb M \,, \T)$ that satisfies the martingale-like conditions given in Definition \ref{martingaletransformationdefinitionmain}. Such transformations enable us to identify $\mathbb M$-valued martingales on $\T$ as collections of boundedness preserving endomorphisms that relate $\boldsymbol{\Lambda}(\mathbb M \,, \T)$-valued random variables to each other under conditional expectations. 
\begin{defn} {\em 
We introduce  the following spaces:
\begin{enumerate}
\item $\mathbb{H}^{(-)}_{\nu,\G}(\boldsymbol{\Lambda}(\M \,, \T))$ : All supermartingale projections on $\boldsymbol{\Lambda}(\M \,, \T)$ with respect to $\G$ under $\nu$-measure.
\item $\mathbb{H}^{(+)}_{\nu,\G}(\boldsymbol{\Lambda}(\M \,, \T))$: All submartingale projections on $\boldsymbol{\Lambda}(\M \,, \T)$ with respect to $\G$ under $\nu$-measure.
\item $\mathbb{H}_{\nu,\G}(\boldsymbol{\Lambda}(\M \,, \T))$: All martingale projections on $\boldsymbol{\Lambda}(\M \,, \T)$ with respect to $\G$ under $\nu$-measure.
\end{enumerate}
}
\end{defn}
\noindent From Definition \ref{martingaletransformationdefinitionmain} and $\preceq$ being a partial-order, we see that
\begin{align}
\mathbb{H}_{\nu,\G}(\boldsymbol{\Lambda}(\mathbb M \,, \T)) = \mathbb{H}^{(-)}_{\nu,\G}(\boldsymbol{\Lambda}(\mathbb M \,, \T)) \,\, \bigcap \,\, \mathbb{H}^{(+)}_{\nu,\G}(\boldsymbol{\Lambda}(\mathbb M \,, \T)). 
\end{align}

\begin{ex}
Set $\M=\R$ and let $\epsilon: \Omega \rightarrow \R$ be a mutually independent random variable such that $\mathbb{E}^{\nu}\left[ \epsilon \right] = 0$. Let $\boldsymbol{\Gamma}^{(p,r)\mapsto(t)}\left(  \boldsymbol{X}_{\T_{[p,r]}} \right) = \inf\{ \boldsymbol{X}_{\T_{[p,r]}} \} + t\epsilon$. Therefore,
\begin{align}
\mathbb{E}^{\nu}\left[\boldsymbol{\Gamma}^{(p,r)\mapsto(t)}\left(  \boldsymbol{X}_{\T_{[p,r]}} \right)  \,\, \left| \,\, \G_{\T_{[p,r]}} \right.\right] = \inf\{ \boldsymbol{X}_{\T_{[p,r]}} \} + t\mathbb{E}^{\nu}\left[\epsilon \right] \preceq X_{r} + 0 \preceq X_{r}, \notag
\end{align}
and hence, $\boldsymbol{\Gamma}^{(p,r)\mapsto(t)} \in \mathbb{H}^{(-)}_{\nu,\G}(\boldsymbol{\Lambda}(\R \,, \T))$.
Similarly,
\begin{enumerate} 
\item Let $\boldsymbol{\Gamma}^{(p,r)\mapsto(t)}\left(  \boldsymbol{X}_{\T_{[p,r]}} \right) = \sup\{ \boldsymbol{X}_{\T_{[p,r]}} \} + t\epsilon$. Then $\boldsymbol{\Gamma}^{(p,r)\mapsto(t)} \in \mathbb{H}^{(+)}_{\nu,\G}(\boldsymbol{\Lambda}(\R \,, \T))$.
\item Let $\boldsymbol{\Gamma}^{(p,r)\mapsto(t)}\left(  \boldsymbol{X}_{\T_{[p,r]}} \right) = X_{r} + t\epsilon$. Then $\boldsymbol{\Gamma}^{(p,r)\mapsto(t)} \in  \mathbb{H}_{\nu,\G}(\boldsymbol{\Lambda}(\R \,, \T))$.
\end{enumerate} 
\end{ex}

The generalization brought through $\mathbb{H}^{(-)}_{\nu,\G}(\boldsymbol{\Lambda}(\mathbb M \,, \T))$, $\mathbb{H}^{(+)}_{\nu,\G}(\boldsymbol{\Lambda}(\mathbb M \,, \T))$ and $\mathbb{H}_{\nu,\G}(\boldsymbol{\Lambda}(\mathbb M \,, \T))$, in relation to classical martingales, will be important when studying quantum decoherence, where we would not necessarily ask the open quantum system to satisfy a property throughout its \emph{entire} lifetime, but rather over possibly compartmentalized system-environment interactions -- this will be clarified in what follows.

\section{Quantum Decoherence}
We shall now focus on open quantum systems as in reference \cite{Brody Hughston 2025}, for which the interactions between quantum systems and their environments can perturb those systems. First, let $\mathcal{Q} = \{1,\ldots,Q\}$ for some $Q\in\N_+$, and consider a quantum system initially in a pure state $\left|\psi_0\right\rangle$ expressed in the form
\begin{align}
\left|\psi_0\right\rangle = \sum_{q\in\mathcal{Q}} \sqrt{\pi_0^{(q)}} \exp\left( \text{i} \theta^{(q)} \right) \left|E^{(q)}\right\rangle, 
\end{align}
where $\pi_0^{(q)}\geq 0$ for all $q\in\mathcal{Q}$. Here $\{ |E^{(q)}\rangle \,: \, q\in \mathcal{Q} \}$ is a set of normalized eigenstates of a non-degenerate observable $\hat{E}$ and $\{ \theta^{(q)} \,: \, q\in \mathcal{Q} \}$ is the set of corresponding phases. In this section, it should be understood that $\sqrt{x}=|\sqrt{x}|$ for any $x\geq 0$. The initial density matrix of the system is $\Psi_0 =  |\psi_0\rangle \langle\psi_0|$, the coordinates of which are given by
\begin{align}
\Psi_0^{(i,j)} = \langle E^{(i)} |\psi_0\rangle \langle \psi_0| E^{(j)} \rangle = \sqrt{\pi_0^{(i)}\pi_0^{(j)}}\exp\left( \text{i}\left( \theta^{(i)} - \theta^{(j)} \right) \right),   \hspace{0.1in} \text{$\forall i,j\in \mathcal{Q}$},
\end{align}
in the $\hat{E}$-basis, so that the overall magnitudes are 
\begin{align}
|\Psi_0^{(i,j)}| = \sqrt{\pi_0^{(i)}\pi_0^{(j)}}, \hspace{0.1in} \text{$\forall i,j\in \mathcal{Q}$}. \label{consistentinitialization}
\end{align}
For what follows, we choose $\T = \{t_0, t_1, \ldots, t_M\}$ as a countable ordered set for some $t_M < \infty$ and $M\in\N_+$, with $t_0 = 0$. The quantum system then interacts with its environment at $t_1\in\T$ and manifests a transformation that generates a new state in the following form:
\begin{align}
\left|\psi_{t_{0}}\right\rangle \mapsto \left|\psi_{t_1}\right\rangle = \sum_{q\in\mathcal{Q}} \sqrt{\pi_{t_1}^{(q)}} \exp\left( \text{i} \theta^{(q)} \right) \left|E^{(q)}\right\rangle, 
\end{align}
at $t_1\in\T$ with $\pi_{t_1}^{(q)}\geq 0$ for all $q\in\mathcal{Q}$, which highlights the fact that we work with phase-homogeneous quantum systems where $\{ \theta^{(q)} \,: \, q\in \mathcal{Q} \}$ is invariant under transformations. As the system continues to interact with its environment, new states arise as a series of random transformations
\begin{align}
\left\{\left|\psi_{t_{0}}\right\rangle, \ldots, \left|\psi_{t_{k-1}}\right\rangle \right\} \mapsto \left|\psi_{t_{k}}\right\rangle = \sum_{q\in\mathcal{Q}} \sqrt{\pi_{t_{k}}^{(q)}} \exp\left( \text{i} \theta^{(q)} \right) \left|E^{(q)}\right\rangle, \label{mappingexpression}
\end{align}
for any $t_{k}\in\T\setminus\{t_0\}$ with $\pi_{t_{k}}^{(q)}\geq 0$ for all $q\in\mathcal{Q}$. 
\begin{rem}{\em
Note that the collection $\{\pi_{t_{k}}^{(q)} \, : \forall q\in\mathcal{Q}\}$ defines a probability distribution of measurement outcomes for every $t_{k}\in\T$. Hence, we have 
\begin{align}
\sum_{i\in\mathcal{Q}}\pi_{t_{k}}^{(i)} = 1, \,\,\,\,\, \forall t_{k}\in\T.\notag 
\end{align}
}
\end{rem}
Since an external observer has no information about the details of such transformations, each probability $\pi_{t_{k}}^{(q)}$ for $t_k\in\T\setminus\{t_0\}$ and $q\in\mathcal{Q}$ is a random variable for that observer. We define the expected state of the system, given previous states, under $\nu$-measure, by setting
\begin{align}
\mathbb{E}^{\nu^{[t_{0},t_{k-1}]}}\left[\Psi_{t_{k}}\right] &= \mathbb{E}^{\nu}\left[ \,\, |\psi_{t_{k}}\rangle \langle\psi_{t_{k}}|  \,\, \left| \,\, \pi_{t_{0}}^{(q)}, \ldots, \pi_{t_{k-1}}^{(q)} \,: \, \forall q\in\mathcal{Q} \, \right.\right] \notag \\
&= \mathbb{E}^{\nu}\left[ \,\, |\psi_{t_{k}}\rangle \langle\psi_{t_{k}}|  \,\, \left| \,\, \left|\psi_{t_{0}}\right\rangle, \ldots, \left|\psi_{t_{k-1}}\right\rangle \right.\right], \hspace{0.1in} \text{$\forall t_{k}\in\T\setminus\{t_0\}$}, 
\end{align}
for $t_k\in\T\setminus\{t_0\}$. Therefore, the corresponding expected coordinates of each density matrix, in the $\hat{E}$-basis, progressing over $\T$ is given by
\begin{align}
\mathbb{E}^{\nu^{[t_{0},t_{k-1}]}}\left[\Psi_{t_{k}}^{(i,j)}\right] &= \mathbb{E}^{\nu}\left[ \,\, \langle E^{(i)} |\psi_{t_{k}}\rangle \langle \psi_{t_{k}}| E^{(j)} \rangle \,\, \left| \,\, \pi_{t_{0}}^{(q)}, \ldots, \pi_{t_{k-1}}^{(q)} \,: \, \forall q\in\mathcal{Q} \, \right.\right] \label{conditionalquantumexpressionfirst} \\
&= \mathbb{E}^{\nu}\left[ \,\,  \sqrt{\pi_{t_{k}}^{(i)}\pi_{t_{k}}^{(j)}}\exp\left( \text{i}\left( \theta^{(i)} - \theta^{(j)} \right) \right) \,\, \left| \pi_{t_{0}}^{(q)}, \ldots, \pi_{t_{k-1}}^{(q)} \,: \, \forall q\in\mathcal{Q} \, \right.\right] \notag \\
&= \exp\left( \text{i}\left( \theta^{(i)} - \theta^{(j)} \right) \right)\mathbb{E}^{\nu}\left[ \,\,  \sqrt{\pi_{t_{k}}^{(i)}\pi_{t_{k}}^{(j)}} \,\, \left| \,\, \left|\psi_{t_{0}}\right\rangle, \ldots, \left|\psi_{t_{k-1}}\right\rangle \right.\right] \label{conditionalquantumexpression}
\end{align}
for every $t_k\in\T\setminus\{t_0\}$ and $i,j\in \mathcal{Q}$. Note that we have
\begin{align}
\mathbb{E}^{\nu^{[t_{0},t_{k}]}}\left[\Psi_{t_{k}}^{(i,j)}\right] &= \mathbb{E}^{\nu}\left[ \,\,  \sqrt{\pi_{t_{k}}^{(i)}\pi_{t_{k}}^{(j)}}\exp\left( \text{i}\left( \theta^{(i)} - \theta^{(j)} \right) \right) \,\, \left| \,\, \pi_{t_{0}}^{(q)}, \ldots, \pi_{t_{k}}^{(q)} \,: \, \forall q\in\mathcal{Q} \, \right.\right] \notag \\
&= \sqrt{\pi_{t_{k}}^{(i)}\pi_{t_{k}}^{(j)}}\exp\left( \text{i}\left( \theta^{(i)} - \theta^{(j)} \right) \right), 
\end{align}
which further yields the overall magnitudes of the density matrix as follows:
\begin{align}
\left|\mathbb{E}^{\nu^{[t_{0},t_{k}]}}\left[\Psi_{t_{k}}^{(i,j)}\right] \right| &= \sqrt{\pi_{t_{k}}^{(i)}\pi_{t_{k}}^{(j)}} \label{expectationremoved} 
\end{align}
for every $t_k\in\T$ and $i,j\in \mathcal{Q}$. Note that (\ref{expectationremoved}) is consistent with (\ref{consistentinitialization}). For what follows, we write $\mathbb{H}$ as the codomain of $\left|\psi_{t_{k}}\right\rangle$ for any $t_{k}\in\T\setminus\{t_0\}$. Accordingly, we have 
\begin{align}
\left|\psi_{t_{0}}\right\rangle \in \mathbb{H} \,\,\,\, \text{and} \,\,\,\, \left|\psi_{t_{k}}\right\rangle : \Omega \times \T \rightarrow \mathbb{H}.  
\end{align}
Since the state space $\mathbb{H}$ is a complex separable Hilbert space, it is a Polish space. Therefore, we can employ projections on $\boldsymbol{\Lambda}(\mathbb{H} \,, \T)$ as path transformations from $(\mathbb{H}\times\T_{[t_0,t_{k-1}]})$ to $(\mathbb{H}\times \{t_{k}\})$ for every $t_{k}\in\T\setminus\{t_0\}$. This allows us to express (\ref{mappingexpression}) as
\begin{align}
\left|\psi_{t_{k}}\right\rangle = \boldsymbol{\Gamma}^{(t_0,t_{k-1})\mapsto(t_k)}\left( \left|\boldsymbol{\psi}_{\T_{[t_0,t_{k-1}]}}\right\rangle \right)  = \sum_{q\in\mathcal{Q}} \sqrt{\pi_{t_{k}}^{(q)}} \exp\left( \text{i} \theta^{(q)} \right) \left|E^{(q)}\right\rangle, \hspace{0.1in} \text{$\forall t_{k}\in\T\setminus\{t_0\}$}, \label{rewritetransformation}
\end{align}
given that $\boldsymbol{\Gamma}^{(t_0,t_{k-1})\mapsto(t_k)}: \boldsymbol{\Lambda}(\mathbb{H} \,, \T) \rightarrow \boldsymbol{\Lambda}(\mathbb{H} \,, \T)$ exists, and where
\begin{align}
\left|\boldsymbol{\psi}_{\T^*}\right\rangle \triangleq \{ \left|\psi_{t}\right\rangle \, : \, \forall t\in\T^* \} \,\, \text{for $\T^*\subseteq\T$} \,\,\,\text{\&}\,\,\, \T^* \neq \emptyset. 
\end{align}
Accordingly, if we define the $\sigma$-algebra
\begin{align}
\G^{\boldsymbol{\psi}}_{\T^*} = \sigma\left( \{\left|\psi_{t}\right\rangle \, : \, \forall t\in\T^*\} \right) = \sigma\left( \left|\boldsymbol{\psi}_{\T^*}\right\rangle \right) \,\, \text{for $\T^*\subseteq\T$}  \,\,\,\text{\&}\,\,\, \T^* \neq \emptyset, 
\end{align}
we can re-write (\ref{conditionalquantumexpression}) as follows:
\begin{align}
\mathbb{E}^{\nu^{[t_{0},t_{k-1}]}}\left[\Psi_{t_{k}}^{(i,j)}\right] = \exp\left( \text{i}\left( \theta^{(i)} - \theta^{(j)} \right) \right)\mathbb{E}^{\nu}\left[ \,\,  \sqrt{\pi_{t_{k}}^{(i)}\pi_{t_{k}}^{(j)}} \,\, \left| \,\, \G^{\boldsymbol{\psi}}_{\T_{[t_0,t_{k-1}]}} \right.\right]. 
\end{align}
From this point onwards, we assume that $\pi_{t_{k}}^{(i)}$ and $\pi_{t_{k}}^{(j)}$ are linearly independent for $i,j\in\mathcal{Q}$, $i\neq j$, for any fixed $t_k\in\T$ in the spirit of \cite{Brody Hughston 2025}. Accordingly, there does \emph{not} exist a set of constants $\alpha^{(i)}$ for $i\in\mathcal{Q}$ such that $\sum_{i\in\mathcal{Q}}\alpha^{(i)}\pi_{t_{k}}^{(i)} = 0$ holds for any fixed $t_k\in\T$. This assumption is to avoid trivial and degenerate cases in the system. As an example, if $\mathcal{Q}=\{1,2\}$, and $\pi_{t_{k}}^{(1)}$ and $\pi_{t_{k}}^{(2)}$ are linearly dependent, then there would exist $\alpha^{(1)}$ and $\alpha^{(2)}$ producing a deterministic relation between $\pi_{t_{k}}^{(1)}$ and $\pi_{t_{k}}^{(2)}$ with a system of linear equations
\begin{align}
\alpha^{(1)}\pi_{t_{k}}^{(1)} + \alpha^{(2)}\pi_{t_{k}}^{(2)} =0 \,\,\,\,\, \text{and} \,\,\,\,\, \pi_{t_{k}}^{(1)} + \pi_{t_{k}}^{(2)} = 1, \notag
\end{align}
which would further imply that $\pi_{t_{k}}^{(1)} = \alpha^{(2)}/(\alpha^{(2)} - \alpha^{(1)})$ and $\pi_{t_{k}}^{(2)} = \alpha^{(1)}/(\alpha^{(1)} - \alpha^{(2)})$ are constants, not random variables. Hence, linear independence excludes non-generic pathological couplings in the system. Next, we introduce the natural $\sigma$-algebras
\begin{align}
\G^{\boldsymbol{\pi}^{(q)}}_{\T^*} = \sigma\left( \{\pi_{t}^{(q)} \, : \, \forall t\in\T^*\} \right) = \sigma\left( \boldsymbol{\pi}_{\T^*}^{(q)} \right) \,\, \text{for $\T^*\subseteq\T$} \,\,\,\text{\&}\,\,\, \T^* \neq \emptyset, 
\end{align}
for every $q\in\mathcal{Q}$, and define the following $\sigma$-algebra for the whole system: 
\begin{align}
\G_{\T^*} \, \triangleq \, \bigvee_{q\in\mathcal{Q}}\G^{\boldsymbol{\pi}^{(q)}}_{\T^*} \label{wholesigmaalgebra}, 
\end{align}
for any $\T^*\subseteq\T$. Finally, we set $\M = \R_+$ (which is a Polish space), since we have
\begin{align}
\pi_{t_{0}}^{(i)} \in \R_+ \,\,\,\, \text{and} \,\,\,\,  \pi_{t_{k}}^{(i)} : \Omega \times \T \rightarrow \R_+, 
\end{align}
for every $i\in\mathcal{Q}$. Hence, the partial order $\preceq$ can be replaced by the total order $\leq$, below.
\begin{prop}
\label{mainproposition}
Let $\left|\psi_{t_{k}}\right\rangle$ admit a path transformation of the form in {\rm (\ref{rewritetransformation})} for a fixed $t_k\in\T\setminus\{t_0\}$. If there exists a projection $\boldsymbol{\Gamma}_{\pi^{(q)}}^{(t_0,t_{k-1})\mapsto(t_k)}: \boldsymbol{\Lambda}(\R_+ \,, \T) \rightarrow \boldsymbol{\Lambda}(\R_+ \,, \T)$ where 
\begin{align}
\pi^{(q)}_{t_{k}} = \boldsymbol{\Gamma}_{\pi^{(q)}}^{(t_0,t_{k-1})\mapsto(t_k)}\left(  \boldsymbol{\pi}_{\T_{[t_0,t_{k-1}]}}^{(q)} \right)  \label{projectionforpi}
\end{align}
and $\boldsymbol{\Gamma}_{\pi^{(q)}}^{(t_0,t_{k-1})\mapsto(t_k)} \in \mathbb{H}^{(-)}_{\nu,\G}(\boldsymbol{\Lambda}(\R_+ \,, \T))$ for every $q\in\mathcal{Q}$, then
\begin{align}
\left| \mathbb{E}^{\nu^{[t_{0},t_{k-1}]}}\left[\Psi_{t_{k}}^{(i,j)}\right] \right| < \left| \Psi_{t_{k-1}}^{(i,j)} \right| \label{decoheremceoffdiagnal}
\end{align}
for every $i,j\in\mathcal{Q}$ where $i\neq j$ at that $t_k\in\T\setminus\{t_0\}$.
\end{prop}
\begin{proof}
Since (\ref{projectionforpi}) holds and 
\begin{align}
\boldsymbol{\Gamma}_{\pi^{(q)}}^{(t_0,t_{k-1})\mapsto(t_k)} \in \mathbb{H}^{(-)}_{\nu,\G}(\boldsymbol{\Lambda}(\R_+ \,, \T))  
\end{align}
for every $q\in\mathcal{Q}$, we have
\begin{align}
\mathbb{E}^{\nu}\left[\boldsymbol{\Gamma}_{\pi^{(q)}}^{(t_0,t_{k-1})\mapsto(t_k)}\left(  \boldsymbol{\pi}_{\T_{[t_0,t_{k-1}]}}^{(q)} \right)  \,\, \left| \,\, \G_{\T_{[t_0,t_{k-1}]}} \right.\right] \leq \pi^{(q)}_{\sup(\T_{[t_0,t_{k-1}]})},\end{align}
which implies that
\begin{align}
\mathbb{E}^{\nu}\left[\pi^{(q)}_{t_{k}}  \,\, \left| \,\, \G_{\T_{[t_0,t_{k-1}]}} \right.\right] \leq \pi^{(q)}_{t_{k-1}}  \label{supermartingalepercforquantum}
\end{align}
for every $q\in\mathcal{Q}$ at that $t_k\in\T\setminus\{t_0\}$. Now denote
\begin{align}
\left|\left| \sqrt{\pi^{(q)}_{t_{k}} }\right|\right|^{\nu^{[t_{0},t_{k-1}]}}_p = \left( \mathbb{E}^{\nu}\left[\left(\sqrt{\pi^{(q)}_{t_{k}}}\right)^p \,\, \left| \,\, \G_{\T_{[t_0,t_{k-1}]}} \right.\right] \right)^{\frac{1}{p}} \hspace{0.1in} \text{for $p\geq 1$}. 
\end{align}
By use of (\ref{supermartingalepercforquantum}) and the monotonicity of $ x \mapsto \sqrt{x}$ we deduce that
\begin{align}
\left|\left| \sqrt{\pi^{(q)}_{t_{k}} }\right|\right|^{\nu^{[t_{0},t_{k-1}]}}_2 \leq \sqrt{\pi^{(q)}_{t_{k-1}}}  \label{pnormexpression}
\end{align}
for every $q\in\mathcal{Q}$ at that $t_k\in\T\setminus\{t_0\}$. Then, using the fact that $\pi_{t_{k}}^{(q)}\geq 0$ for all $q\in\mathcal{Q}$, employing (\ref{conditionalquantumexpressionfirst})-(\ref{conditionalquantumexpression}) and using (\ref{wholesigmaalgebra}), we can write
\begin{align}
\left|\mathbb{E}^{\nu^{[t_{0},t_{k-1}]}}\left[\Psi_{t_{k}}^{(i,j)}\right] \right| &= \mathbb{E}^{\nu}\left[ \,\,  \sqrt{\pi_{t_{k}}^{(i)}\pi_{t_{k}}^{(j)}} \,\, \left| \G_{\T_{[t_0,t_{k-1}]}} \right.\right]. \label{absolutedensitymatrixconditional}
\end{align}
Since $\pi_{t_{k}}^{(q)}$ are linearly independent for $q\in\mathcal{Q}$, by use of conditional H\"older's inequality, we have
\begin{align}
\mathbb{E}^{\nu}\left[ \,  \sqrt{\pi_{t_{k}}^{(i)}\pi_{t_{k}}^{(j)}} \,\, \left| \G_{\T_{[t_0,t_{k-1}]}} \right.\right] < \left|\left| \sqrt{\pi^{(i)}_{t_{k}} }\right|\right|^{\nu^{[t_{0},t_{k-1}]}}_2 \left|\left| \sqrt{\pi^{(j)}_{t_{k}} }\right|\right|^{\nu^{[t_{0},t_{k-1}]}}_2 
\end{align}
for $i,j\in\mathcal{Q}$, $i\neq j$, which, by use of (\ref{pnormexpression}), further implies
 \begin{align}
\mathbb{E}^{\nu}\left[ \,  \sqrt{\pi_{t_{k}}^{(i)}\pi_{t_{k}}^{(j)}} \,\, \left| \G_{\T_{[t_0,t_{k-1}]}} \right.\right] < \sqrt{\pi^{(i)}_{t_{k-1}} \pi^{(j)}_{t_{k-1}}} 
= \left|\mathbb{E}^{\nu^{[t_{0},t_{k-1}]}}\left[\Psi_{t_{k-1}}^{(i,j)}\right] \right| \label{finalinequality}
\end{align}
for $i,j\in\mathcal{Q}$, $i\neq j$, where the last equality follows from (\ref{expectationremoved}). Therefore, by use of (\ref{absolutedensitymatrixconditional}) and (\ref{finalinequality}), we see that
\begin{align}
\left|\mathbb{E}^{\nu^{[t_{0},t_{k-1}]}}\left[\Psi_{t_{k}}^{(i,j)}\right] \right| < \left|\mathbb{E}^{\nu^{[t_{0},t_{k-1}]}}\left[\Psi_{t_{k-1}}^{(i,j)}\right] \right| = \left| \Psi_{t_{k-1}}^{(i,j)} \right| 
\end{align}
holds for every $i,j\in\mathcal{Q}$ where $i\neq j$ at that $t_k\in\T\setminus\{t_0\}$.
\end{proof}
Proposition \ref{mainproposition} shows that the off-diagonal magnitudes of the expected density matrix are smaller than those of the most recently observed density matrix of the given open quantum system. Hence, the existence of a supermartingale projection on the diagonal magnitudes of an observed ensemble of density matrices implies an expected decoherence in that system. Now, in the spirit of \cite{Brody Hughston 2025}, we continue by studying the Shannon-Wiener information given by
\begin{align}
\mathbb{E}^{\nu^{[t_{0},t_{k-1}]}}\left[\mathcal{S}_{t_{k}}\right] &= \mathbb{E}^{\nu}\left[ \,\, \sum_{q\in\mathcal{Q}} \pi_{t_{k}}^{(q)} \log\left( \pi_{t_{k}}^{(q)} \right) \,\, \left| \,\, \pi_{t_{0}}^{(q)}, \ldots, \pi_{t_{k-1}}^{(q)} \,: \, \forall q\in\mathcal{Q} \, \right.\right], \label{expectedshannonentropyexpression}
\end{align}
for every $t_k\in\T\setminus\{t_0\}$, which is essentially the negative of the Shannon-Wiener entropy. 
Accordingly, we have
\begin{align}
\mathbb{E}^{\nu^{[t_{0},t_{k-1}]}}\left[\mathcal{S}_{t_{k-1}}\right] &= \mathbb{E}^{\nu}\left[ \,\, \sum_{q\in\mathcal{Q}} \pi_{t_{k-1}}^{(q)} \log\left( \pi_{t_{k-1}}^{(q)} \right) \,\, \left| \,\, \pi_{t_{0}}^{(q)}, \ldots, \pi_{t_{k-1}}^{(q)} \,: \, \forall q\in\mathcal{Q} \, \right.\right]. \label{initialshannonentropyexpressionfirststatement}
\end{align}
Note that the $\sigma$-algebra in (\ref{initialshannonentropyexpressionfirststatement}) generated by $\{\pi_{t_{0}}^{(q)}, \ldots, \pi_{t_{k-1}}^{(q)} \,: \, \forall q\in\mathcal{Q}\}$ renders each $\pi_{t_{k-1}}^{(q)} \log\left( \pi_{t_{k-1}}^{(q)} \right)$ measurable for $q\in\mathcal{Q}$ under conditional expectation. Therefore, we can write as follows:
\begin{align}
\mathbb{E}^{\nu^{[t_{0},t_{k-1}]}}\left[\mathcal{S}_{t_{k-1}}\right] &=\sum_{q\in\mathcal{Q}} \pi_{t_{k-1}}^{(q)} \log\left( \pi_{t_{k-1}}^{(q)} \right). \label{initialshannonentropyexpressionsecondstatement}
\end{align}
Finally, we see that the right-hand side of (\ref{initialshannonentropyexpressionsecondstatement}) is Shannon-Wiener information at time $t_{k-1}\in\T$, which allows us to write (\ref{initialshannonentropyexpressionfirststatement}) as
\begin{align}
\mathbb{E}^{\nu^{[t_{0},t_{k-1}]}}\left[\mathcal{S}_{t_{k-1}}\right] = \mathcal{S}_{t_{k-1}}. \label{initialshannonentropyexpression}
\end{align}
We can now present the next result to evaluate the direction of information in the open quantum system.
\begin{prop}
\label{mainpropositiontwo}
Let $\left|\psi_{t_{k}}\right\rangle$ admit a path transformation of the form in {\rm (\ref{rewritetransformation})} for a fixed $t_k\in\T\setminus\{t_0\}$. If there exists a projection $\boldsymbol{\Gamma}_{\pi^{(q)}}^{(t_0,t_{k-1})\mapsto(t_k)}: \boldsymbol{\Lambda}(\R_+ \,, \T) \rightarrow \boldsymbol{\Lambda}(\R_+ \,, \T)$ where 
\begin{align}
\pi^{(q)}_{t_{k}} = \boldsymbol{\Gamma}_{\pi^{(q)}}^{(t_0,t_{k-1})\mapsto(t_k)}\left(  \boldsymbol{\pi}_{\T_{[t_0,t_{k-1}]}}^{(q)} \right)  \label{projectionforpitwo}
\end{align}
and $\boldsymbol{\Gamma}_{\pi^{(q)}}^{(t_0,t_{k-1})\mapsto(t_k)} \in \mathbb{H}^{(+)}_{\nu,\G}(\boldsymbol{\Lambda}(\R_+ \,, \T))$ for every $q\in\mathcal{Q}$, then
\begin{align}
\mathcal{S}_{t_{k-1}} < \mathbb{E}^{\nu^{[t_{0},t_{k-1}]}}\left[\mathcal{S}_{t_{k}}\right]  \label{shannondecrease}
\end{align}
at that $t_k\in\T\setminus\{t_0\}$.
\end{prop}
\begin{proof}
Since (\ref{projectionforpitwo}) holds and 
\begin{align}
\boldsymbol{\Gamma}_{\pi^{(q)}}^{(t_0,t_{k-1})\mapsto(t_k)} \in \mathbb{H}^{(+)}_{\nu,\G}(\boldsymbol{\Lambda}(\R_+ \,, \T))  
\end{align}
for every $q\in\mathcal{Q}$, we have the relation
\begin{align}
\pi^{(q)}_{\sup(\T_{[t_0,t_{k-1}]})} \leq \mathbb{E}^{\nu}\left[\boldsymbol{\Gamma}_{\pi^{(q)}}^{(t_0,t_{k-1})\mapsto(t_k)}\left(  \boldsymbol{\pi}_{\T_{[t_0,t_{k-1}]}}^{(q)} \right)  \,\, \left| \,\, \G_{\T_{[t_0,t_{k-1}]}} \right.\right],  
\end{align}
which implies that
\begin{align}
  \pi^{(q)}_{t_{k-1}}  \leq \mathbb{E}^{\nu}\left[\pi^{(q)}_{t_{k}}  \,\, \left| \,\, \G_{\T_{[t_0,t_{k-1}]}} \right.\right]  \label{supermartingalepercforquantumtwo}
\end{align}
for every $q\in\mathcal{Q}$ at that $t_k\in\T\setminus\{t_0\}$. In addition, since $x \mapsto \log(x)$ is monotonic, by use of (\ref{supermartingalepercforquantumtwo}) we have 
\begin{align}
\log\left( \pi^{(q)}_{t_{k-1}}\right)  \leq \log\left(\mathbb{E}^{\nu}\left[\pi^{(q)}_{t_{k}}  \,\, \left| \,\, \G_{\T_{[t_0,t_{k-1}]}} \right.\right]\right),  \label{supermartingalepercforquantumtwologarithm}
\end{align}
for every $q\in\mathcal{Q}$ at that $t_k\in\T\setminus\{t_0\}$. Since $x \mapsto x\log(x)$ is strictly convex, using the conditional form of Jensen's inequality, we see that
\begin{align}
\mathbb{E}^{\nu}\left[\pi^{(q)}_{t_{k}}  \,\, \left| \,\, \G_{\T_{[t_0,t_{k-1}]}} \right.\right]\log\left(\mathbb{E}^{\nu}\left[\pi^{(q)}_{t_{k}}  \,\, \left| \,\, \G_{\T_{[t_0,t_{k-1}]}} \right.\right]\right) < \mathbb{E}^{\nu}\left[\pi^{(q)}_{t_{k}}\log\left(\pi^{(q)}_{t_{k}} \right)  \,\, \left| \,\, \G_{\T_{[t_0,t_{k-1}]}} \right.\right]  \label{jensensinequality}
\end{align}
for every $q\in\mathcal{Q}$ at that $t_k\in\T\setminus\{t_0\}$. Thus, using (\ref{supermartingalepercforquantumtwo}), (\ref{supermartingalepercforquantumtwologarithm}) and (\ref{jensensinequality}), we obtain
\begin{align}
\pi^{(q)}_{t_{k-1}}\log\left( \pi^{(q)}_{t_{k-1}}\right) < \mathbb{E}^{\nu}\left[\pi^{(q)}_{t_{k}}\log\left(\pi^{(q)}_{t_{k}} \right)  \,\, \left| \,\, \G_{\T_{[t_0,t_{k-1}]}} \right.\right]  \label{jensensinequalitytwo}
\end{align}
for every $q\in\mathcal{Q}$ at that $t_k\in\T\setminus\{t_0\}$. Hence, from (\ref{initialshannonentropyexpression}) and (\ref{jensensinequalitytwo}), we have
\begin{align}
\mathbb{E}^{\nu^{[t_{0},t_{k-1}]}}\left[\mathcal{S}_{t_{k-1}}\right] = \sum_{q\in\mathcal{Q}} \pi^{(q)}_{t_{k-1}}\log\left( \pi^{(q)}_{t_{k-1}}\right) < \sum_{q\in\mathcal{Q}}\mathbb{E}^{\nu}\left[\pi^{(q)}_{t_{k}}\log\left(\pi^{(q)}_{t_{k}} \right)  \,\, \left| \,\, \G_{\T_{[t_0,t_{k-1}]}} \right.\right] . \label{shannonfirstexpectation}
\end{align}
Since $\mathbb{E}^{\nu}[\,\cdot \,]$ is a linear operator, we can write
\begin{align}
\sum_{q\in\mathcal{Q}}\mathbb{E}^{\nu}\left[\pi^{(q)}_{t_{k}}\log\left(\pi^{(q)}_{t_{k}} \right)  \,\, \left| \,\, \G_{\T_{[t_0,t_{k-1}]}} \right.\right]  &= \mathbb{E}^{\nu}\left[\sum_{q\in\mathcal{Q}}\pi^{(q)}_{t_{k}}\log\left(\pi^{(q)}_{t_{k}} \right)  \,\, \left| \,\, \G_{\T_{[t_0,t_{k-1}]}} \right.\right] \notag \\
&= \mathbb{E}^{\nu^{[t_{0},t_{k-1}]}}\left[\mathcal{S}_{t_{k}}\right]. \label{shannonfirstexpectationtwo}
\end{align}
by employing (\ref{expectedshannonentropyexpression}). Therefore, from (\ref{shannonfirstexpectation}) and (\ref{shannonfirstexpectationtwo}), we have
\begin{align}
\mathbb{E}^{\nu^{[t_{0},t_{k-1}]}}\left[\mathcal{S}_{t_{k-1}}\right] < \mathbb{E}^{\nu^{[t_{0},t_{k-1}]}}\left[\mathcal{S}_{t_{k}}\right] 
\end{align}
at that $t_k\in\T\setminus\{t_0\}$. The claimed inequality \eqref{shannondecrease} then follows from (\ref{initialshannonentropyexpression}).
\end{proof}
If we take Shannon-Wiener information as a measure of negative uncertainty in the system (since it is negative entropy), then Proposition \ref{mainpropositiontwo} shows that the existence of a submartingale projection on the diagonal magnitudes of an observed ensemble of density matrices implies expected information gain in that system, given that Shannon-Wiener information is expected to increase as it progresses over $\T$. 
\begin{rem}
\label{remintersection} {\em
Note that $\mathbb{H}^{(-)}_{\nu,\G}(\boldsymbol{\Lambda}(\mathbb M \,, \T))$ encodes non-increasing behaviour from {\rm (\ref{supermartingaleperc})} and $\mathbb{H}^{(+)}_{\nu,\G}(\boldsymbol{\Lambda}(\mathbb M \,, \T))$ encodes non-decreasing behaviour from {\rm (\ref{submartingaleperc})}. Accordingly, $\mathbb{H}^{(-)}_{\nu,\G}(\boldsymbol{\Lambda}(\R_+ \,, \T))$ manifests decreasing off-diagonals and $\mathbb{H}^{(+)}_{\nu,\G}(\boldsymbol{\Lambda}(\R_+ \,, \T))$ manifests increasing information. }
\end{rem}
Remark \ref{remintersection} presents an intuitive insight as to how supermartingale projections relate to decreasing quantum phenomenon and how submartingale projections relate to increasing quantum phenomenon in open systems. This further sheds light on a connection between superharmonic functions from potential theory and quantum decoherence, as well as between subharmonic functions and information gain -- the details of which we leave as a project for the future.
\begin{prop}
\label{mainpropositiongeneralcombine}
Let $\left|\psi_{t_{k}}\right\rangle$ admit a path transformation of the form in {\rm(\ref{rewritetransformation})} for a fixed $t_k\in\T\setminus\{t_0\}$. If there exists a projection $\boldsymbol{\Gamma}_{\pi^{(q)}}^{(t_0,t_{k-1})\mapsto(t_k)}: \boldsymbol{\Lambda}(\R_+ \,, \T) \rightarrow \boldsymbol{\Lambda}(\R_+ \,, \T)$ where 
\begin{align}
\pi^{(q)}_{t_{k}} = \boldsymbol{\Gamma}_{\pi^{(q)}}^{(t_0,t_{k-1})\mapsto(t_k)}\left(  \boldsymbol{\pi}_{\T_{[t_0,t_{k-1}]}}^{(q)} \right) \label{finalpi}
\end{align}
and $\boldsymbol{\Gamma}_{\pi^{(q)}}^{(t_0,t_{k-1})\mapsto(t_k)} \in \mathbb{H}_{\nu,\G}(\boldsymbol{\Lambda}(\R_+ \,, \T))$ for every $q\in\mathcal{Q}$, then
\begin{align}
\left| \mathbb{E}^{\nu^{[t_{0},t_{k-1}]}}\left[\Psi_{t_{k}}^{(i,j)}\right] \right| < \left| \Psi_{t_{k-1}}^{(i,j)} \right| \label{finaldecoherence}
\end{align}
for every $i,j\in\mathcal{Q}$ where $i\neq j$, and 
\begin{align}
\mathcal{S}_{t_{k-1}} < \mathbb{E}^{\nu^{[t_{0},t_{k-1}]}}\left[\mathcal{S}_{t_{k}}\right] \label{finalentropy}
\end{align}
at that $t_k\in\T\setminus\{t_0\}$.
\end{prop}
\begin{proof}
Since (\ref{finalpi}) holds and 
\begin{align}
\boldsymbol{\Gamma}_{\pi^{(q)}}^{(t_0,t_{k-1})\mapsto(t_k)} \in \mathbb{H}_{\nu,\G}(\boldsymbol{\Lambda}(\R_+ \,, \T)) = \mathbb{H}^{(-)}_{\nu,\G}(\boldsymbol{\Lambda}(\R_+ \,, \T)) \,\, \bigcap \,\, \mathbb{H}^{(+)}_{\nu,\G}(\boldsymbol{\Lambda}(\R_+ \,, \T)) 
\end{align}
for every $q\in\mathcal{Q}$ at that $t_k\in\T\setminus\{t_0\}$, the inequalities (\ref{finaldecoherence}) and (\ref{finalentropy}) follow from Propositions \ref{mainproposition} and \ref{mainpropositiontwo}, respectively.
\end{proof}
Proposition \ref{mainpropositiongeneralcombine} demonstrates that the existence of a martingale projection on the diagonal magnitudes of an observed ensemble of density matrices implies both expected decoherence and information gain in that system. In all the results thus far, we highlight statements of the form \emph{at that $t_k\in\T\setminus\{t_0\}$}, since super/sub-martingale projections may not exist for every interaction between the system and the environment -- this creates flexibility in compartmentalizing quantum interactions across time, instead of asking such features from the \emph{entire} process. 
Hence, our results show that if/when a system-environment interaction manifests a martingale projection on the density matrix, then one can expect both decoherence and information gain as an outcome of that very interaction. If such features do exist for the entire process, then we obtain the following:
\begin{coro}
\label{martingalestatement}
Let $\left|\psi_{t_{k}}\right\rangle$ admit a path transformation of the form in {\rm (\ref{rewritetransformation})} for every $t_k\in\T\setminus\{t_0\}$. If $\{\pi^{(q)}_{t_{k}}\}_{t_k\in\T}$ is a $(\nu,\G)$-supermartingale for every $q\in\mathcal{Q}$, then
\begin{align}
\left| \mathbb{E}^{\nu^{[t_{0},t_{k-1}]}}\left[\Psi_{t_{k}}^{(i,j)}\right] \right| < \left| \Psi_{t_{k-1}}^{(i,j)} \right| \label{finaldecoherencemartingale}
\end{align}
for every $i,j\in\mathcal{Q}$ where $i\neq j$ for every $t_k\in\T\setminus\{t_0\}$. If $\{\pi^{(q)}_{t_{k}}\}_{t_k\in\T}$ is a $(\nu,\G)$-submartingale for every $q\in\mathcal{Q}$, then 
\begin{align}
\mathcal{S}_{t_{k-1}} < \mathbb{E}^{\nu^{[t_{0},t_{k-1}]}}\left[\mathcal{S}_{t_{k}}\right] \label{finalentropymartingale}
\end{align}
for every $t_k\in\T\setminus\{t_0\}$. If $\{\pi^{(q)}_{t_{k}}\}_{t_k\in\T}$ is a $(\nu,\G)$-martingale for every $q\in\mathcal{Q}$, then both {\rm (\ref{finaldecoherencemartingale})} and {\rm (\ref{finalentropymartingale})} hold for every $t_k\in\T\setminus\{t_0\}$.
\end{coro}
Corollary \ref{martingalestatement} places supermartingales at a central position in understanding open quantum systems expected to decohere throughout their entire evolution as a consequence of their interactions with their environments. In addition, Corollary \ref{martingalestatement} shows how submartingales play a role in evaluating the expected direction of information on a continual basis. Accordingly, we see that martingales \emph{glue} a persistent relationship between decoherence and information gain dynamics for open quantum systems. 

\begin{rem}
\label{finalremarkconnection}{\em
Remark \ref{stochasticassociation} and Corollary \ref{martingalestatement} connect when martingales can be identified as collections of law-consistent path transformations; e.g. a family of martingale projections $\mathcal{P}_{(q)}^{(t_0,t_0)}$ exists that generates a $(\nu,\G)$-martingale $\{\pi^{(q)}_{t_{k}}\}_{t_k\in\T}$ from $\pi^{(q)}_{t_{0}}$ for $q\in\mathcal{Q}$.}
\end{rem}

We note in particular that martingales play a pivotal role in energy-based state reduction dynamics in quantum measurement theory \cite{Gisin 1989, Adler et al 2001, Brody Hughston 2002, Menguturk 2016, Brody Hughston 2018, Menguturk Menguturk 2020, Brody Hughston 2023, Menguturk 2024}. Accordingly, we close this chapter by drawing attention to the following observation. In the theory of quantum measurement, martingales have been interpreted in the references cited as randomized analogues for the conservation law of energy in expectation. In the theory of quantum decoherence, martingales can again be interpreted as randomized analogues for the conservation law of energy in expectation. This remark is suggestive of a potentially deeper consolidation of quantum measurement and quantum decoherence through martingale constructs, which we hope to explore in more detail elsewhere. Herein, we shall nonetheless provide a brief elaboration of this point and highlight the existence of martingale processes $\{\bar{\pi}^{(q)}_{t}\}_{t\in[0,T]}$ in stochastic \Schrodinger evolutions in relation to quantum states represented as in (\ref{mappingexpression}). More specifically, denoting $\bar{\T} = [0,T]$ for $T < \infty$, each $\{\bar{\pi}^{(q)}_{t}\}_{t\in\bar{\T}}$ satisfies a stochastic differential equation given by 
\begin{align}
\dd\bar{\pi}^{(q)}_{t} = \sigma^{(q)}_{t}\bar{\pi}^{(q)}_{t}\dd W^{(q)}_t, \label{sdeformeasure}
\end{align}
for $q\in\mathcal{Q}$ and $t\in[0,T)$, where $\{W^{(q)}_t\}_{t\in\T}$ is a standard $\nu$-Brownian motion, and the detailed mathematical expressions for $\{\sigma^{(q)}_t\}_{t\in[0,T)}$ can be found in the references cited above. Note that (\ref{sdeformeasure}) has a closed-form analytical solution given by the following:
\begin{align}
\label{exponentialmartingalerep}
\bar{\pi}^{(q)}_{t} = \bar{\pi}^{(q)}_{0}\exp\left(- \frac{1}{2}\int_0^t \left(\sigma^{(q)}_{s}\right)^2\dd s + \int_0^t \sigma^{(q)}_{s}\dd W^{(q)}_s \right),
\end{align}
for $q\in\mathcal{Q}$ and $t\in[0,T)$. This can also be viewed as a family of martingale projections of $\bar{\pi}^{(q)}_{0}$ in the spirit of Remark \ref{finalremarkconnection} -- accordingly, (\ref{exponentialmartingalerep}) for $t\in[0,T)$ can be considered a collection of path transformations implicitly linked to the dynamical representation (\ref{sdeformeasure}). Thus, using (\ref{exponentialmartingalerep}), we can write
\begin{align}
\boldsymbol{\Gamma}^{(0,u)\mapsto(t)}\left(  \boldsymbol{\bar{\pi}}^{(q)}_{\bar{\T}_{[0,u]}} \right) = \bar{\pi}^{(q)}_{u}\exp\left(- \frac{1}{2}\int_u^t \left(\sigma^{(q)}_{s}\right)^2\dd s + \int_u^t \sigma^{(q)}_{s}\dd W^{(q)}_s \right), \notag
\end{align}
for $q\in\mathcal{Q}$ and every $u\leq t\in[0,T)$. Therefore, we have the martingale projection criterion
\begin{align}
\mathbb{E}^{\nu}\left[\boldsymbol{\Gamma}^{(0,u)\mapsto(t)}\left(  \boldsymbol{\bar{\pi}}^{(q)}_{\bar{\T}_{[0,u]}} \right)  \,\, \left| \,\, \G^{\bar{\pi}^{(q)}}_{\bar{\T}_{[0,u]}} \right.\right] = \bar{\pi}^{(q)}_{u}, \notag
\end{align}
for every $u\leq t\in[0,T)$. Moreover, if we define the Shannon-Wiener information process associated to $\{\bar{\pi}^{(q)}_{t}\}_{t\in\bar{\T}}$ by
\begin{align}
\bar{S}_t = \sum_{q\in\mathcal{Q}}\bar{\pi}^{(q)}_{t} \log\left( \bar{\pi}^{(q)}_{t} \right), \label{newcontinuousshannon}
\end{align}
for $t\in\bar{\T}$, as we have done for $\{\pi^{(q)}_{t_{k}}\}_{t_k\in\T}$ in (\ref{expectedshannonentropyexpression}), one can show that $\{\bar{S}_t\}_{t\in[0,T)}$ is increasing on average by using (\ref{exponentialmartingalerep}) and applying It\^o's formula to (\ref{newcontinuousshannon}) -- for which the detailed proof is beyond the scope of this paper -- where the following stochastic differential equation for $\{\bar{S}_t\}_{t\in[0,T)}$ arises
\begin{align}
\dd\bar{S}_t = \bar{\mu}_{t}\dd t + \bar{\gamma}_{t}\dd \bar{W}_t, \notag
\end{align}
where $\{\bar{W}_t\}_{t\in\T}$ is a standard $\nu$-Brownian motion and $\bar{\mu}_t > 0$ for every $t\in[0,T)$. Note that $\{\bar{S}_t\}_{t\in[0,T)}$ behaves in agreement with (\ref{finalentropymartingale}) from Corollary \ref{martingalestatement}. Accordingly, since $\{\bar{\pi}^{(q)}_{t}\}_{t\in\bar{\T}}$ from (\ref{exponentialmartingalerep}) is a martingale for every $q\in\mathcal{Q}$ and the Shannon-Wiener information process $\{\bar{S}_{t}\}_{t\in\bar{\T}}$ from (\ref{newcontinuousshannon}) is increasing on average, martingale constructs hint towards a unified affiliation between quantum state reduction, quantum decoherence and information gain. 

\section{Conclusion}
It may be helpful if we comment briefly on the significance of the present work. In this paper, we have constructed what we call super/sub-martingale projections as random endomorphisms on path-valued random variables. As part of this objective, we have introduced path transformations on Polish spaces, which form a general framework that allows us to host super/sub-martingale projections within a broader landscape. We envision that the mathematical setting presented concisely in this work can foster further such rigorous analysis in future.

The main goal of this paper has been to study decoherence and information-theoretic behaviour of open quantum systems, which we have approached through super/sub-martingale projections. We have shown that the existence of martingale projections in an open quantum system implies both expected decoherence and expected information gain in that system. 

Towards deepening our understanding of quantum decoherence, supermartingale projections enable us to identify a fundamental condition on system-environment interactions as path-transformations, since they naturally induce expected non-increasing behaviour over possibly compartmentalised time segments. Consequentially, these projections allow us to connect classical supermartingales from stochastic analysis with quantum decoherence under a specific setup. In the same vein, submartingale projections (and in special cases classical submartingales) enable us to connect non-decreasing behaviour with the direction of information in open quantum systems. Accordingly, this paper places martingale projections, and thereby classical martingales, in a central location to merge a deep relation between quantum decoherence and information.

Finally, we discussed how martingales can provide a consolidated connection between quantum state reduction, quantum decoherence and information gain. Accordingly, by interpreting martingale constructs as manifestations of the conservation law of energy in expectation, we may further enrich our understanding of these phenomena in terms of a fundamental principle in physics.

\appendix
\section{Appendix}
Since several involved notations appear throughout the paper, we provide a table of key symbols for the convenience of the reader.
\begin{table}[ht]
\centering
\begin{tabular}{ | p{2cm} | p{4cm} |}
\hline
\,\,\,\textbf{Notation} & \,\,\,\,\,\,\,\, \textbf{Explanation} \\
\hline
\,\,\,\,\,\,\,\,\,\,  $\M$ & \,\,\,\,\,\,\,\, Polish space \\
\hline
\,\,\,\,\,\,\,\,\,\,  $\T$ & \,\,\,\,\,\,\,\, time horizon \\
\hline
\,\,\,\,\,\,\,\,\, $\boldsymbol{X}_{\T}$ & \,\,\,\,\,\,\,\, $\M$-valued path \\
\hline
\,\, $\boldsymbol{\Lambda}(\mathbb M \,, \T)$ & \,\,\,\,\,\,\,\, path space \\
\hline
\,\,\,\,\,\,\,\,\,\, $\G_{\T}$ & \,\,\,\,\,\,\,\,\,\,\, $\sigma$-algebra \\
\hline
\,\,\,\,\,\,\,\,\,\, $\nu$ & \, probability measure \\
\hline
\,\,\,\,\,\,\,\,\,\, $\mathbb{E}$ & \, expectation operator \\
\hline
\, $\boldsymbol{\Gamma}^{(p,r)\mapsto(s,t)}$ & \, path transformation \\
\hline
\end{tabular}
\end{table}
\\
We shall also give some examples of the different properties that path transformations from Definition \ref{transformationsgeneral} can display. We note that the examples below are far from being exhaustive.
\begin{ex}
\emph{For $p\leq r \in \T$, if a path transformation satisfies
\begin{align}
\boldsymbol{X}_{\T_{[p,r]}} = \boldsymbol{\Gamma}^{(p,r)\mapsto(p,r)}\left( \boldsymbol{X}_{\T_{[p,r]}}\right), 
\end{align}
then $\boldsymbol{\Gamma}^{(p,r)\mapsto(p,r)}$ is a path-preserving transformation on $\boldsymbol{\Lambda}(\mathbb M \,, \T)$; an identity-transformation.}
\end{ex}

\begin{ex}
\emph{For $p\leq r \in \T$ and $s\leq t \in \T$, if a path transformation satisfies
\begin{align}
\boldsymbol{\Gamma}^{(p,r)\mapsto(s,t)}\left( \boldsymbol{X}_{\T_{[p,r]}}\right) = \boldsymbol{\Gamma}^{(u,v)\mapsto(s,t)}\left( \boldsymbol{X}_{\T_{[u,v]}}\right), \hspace{0.1in} \text{for $\T_{[u,v]}\subset \T_{[p,r]}$} 
\end{align}
then $\boldsymbol{\Gamma}^{(p,r)\mapsto(s,t)}$ is a reducible path transformation on $\boldsymbol{\Lambda}(\mathbb M \,, \T)$ from $\T_{[p,r]}$ onto $\T_{[u,v]}$.}
\end{ex}

\begin{ex}
\emph{For $p\leq r \in \T$ and $s\leq t \in \T$, if an endomorphism 
\begin{align}
\left(\boldsymbol{\Gamma}^{(p,r)\mapsto(s,t)}\right)^{*} = \boldsymbol{\Gamma}_{*}^{(s,t)\mapsto(p,r)} \,: \, \boldsymbol{\Lambda}(\mathbb M \,, \T) \rightarrow \boldsymbol{\Lambda}(\mathbb M \,, \T) 
\end{align}
exists such that
\begin{align}
\boldsymbol{Y}_{\T_{[s,t]}} = \boldsymbol{\Gamma}^{(p,r)\mapsto(s,t)}\left( \boldsymbol{X}_{\T_{[p,r]}}\right) \, \implies \, \boldsymbol{X}_{\T_{[p,r]}} = \boldsymbol{\Gamma}_{*}^{(s,t)\mapsto(p,r)}\left(\boldsymbol{Y}_{\T_{[s,t]}}\right), 
\end{align}
then $\boldsymbol{\Gamma}^{(p,r)\mapsto(s,t)}$ is an invertible path transformation on $\boldsymbol{\Lambda}(\mathbb M \,, \T)$ between $\T_{[p,r]}$ and $\T_{[s,t]}$.}
\end{ex}

\begin{ex}
\emph{We present two path transformations directly motivated by the path-operations of \cite{Dupire 2009, Cont Fournie 2010} used in functional-It\^o calculus. The first path operation is the \emph{vertical bump} given by
\begin{equation}
\label{dupirevertical}
    \boldsymbol{Y}_{\T_{u,t}}(s) =
\begin{cases}
      X_s, & \text{if}\ s\in [u,t) \\
      X_t + \epsilon, & \text{if}\ s = t,
\end{cases}
\end{equation}
for some $\epsilon > 0$, with $\boldsymbol{X}_{\T_{u,t}}\in \boldsymbol{\Lambda}(\R , \T)$ and $\boldsymbol{Y}_{\T_{u,t}}(s)\in\R$ is the value of $\boldsymbol{Y}_{\T_{u,t}}\in \boldsymbol{\Lambda}(\R , \T)$ at point $s\in\T_{u,t}$. The second path operation is the \emph{horizontal stretch} given by
\begin{equation}
\label{dupirehext}
    \boldsymbol{Z}_{\T_{u, t + \alpha}}(s) =
\begin{cases}
			X_s, & \text{if}\ s\in [u,t] \\
      X_t, & \text{if}\ s\in(t,t+\alpha],
\end{cases}
\end{equation}
for some $\alpha> 0$, such that
\begin{align}
t+\alpha \in \T \,\,\, \Rightarrow \,\,\, \boldsymbol{Z}_{\T_{u,t+\alpha}}\in \boldsymbol{\Lambda}(\R , \T). 
\end{align}
The operations in (\ref{dupirevertical}) and (\ref{dupirehext}) form the building blocks of functional-derivative operators in the sense of Dupire \cite{Dupire 2009}. Using Definition \ref{transformationsgeneral}, one can define $\boldsymbol{\Gamma}_{\epsilon}$ and $\boldsymbol{\Gamma}_{\alpha}$ with
\begin{align}
\boldsymbol{Y}_{\T_{u,t}} = \boldsymbol{\Gamma}_{\epsilon}^{(u,t)\mapsto(u,t)}\left( \boldsymbol{X}_{\T_{[u,t]}}\right) \hspace{0.1in} \text{and} \hspace{0.1in}\boldsymbol{Z}_{\T_{u,t+\alpha}} = \boldsymbol{\Gamma}_{\alpha}^{(u,t)\mapsto(u,t+\alpha)}\left( \boldsymbol{X}_{\T_{[u,t]}}\right) 
\end{align}
as the \emph{vertical bump} in (\ref{dupirevertical}), and the \emph{horizontal stretch} in (\ref{dupirehext}), respectively. We should highlight that
\begin{align}
\left(  \boldsymbol{\Gamma}_{\alpha}^{(u,t)\mapsto(u,t+\alpha)}  \circ  \boldsymbol{\Gamma}_{\epsilon}^{(u,t)\mapsto(u,t)} \right) \left( \boldsymbol{X}_{\T_{[u,t]}}\right) \neq \left( \boldsymbol{\Gamma}_{\epsilon}^{(u,t+\alpha)\mapsto(u,t+\alpha)}  \circ  \boldsymbol{\Gamma}_{\alpha}^{(u,t)\mapsto(u,t+\alpha)} \right) \left( \boldsymbol{X}_{\T_{[u,t]}}\right)
\end{align}
for a fixed $\boldsymbol{X}_{\T_{u,t}}\in \boldsymbol{\Lambda}(\R , \T)$, and therefore, $\boldsymbol{\Gamma}_{\epsilon}$ and $\boldsymbol{\Gamma}_{\alpha}$ are \emph{non-commutative} when defined as in (\ref{dupirevertical}) and (\ref{dupirehext}), respectively.}
\end{ex}

\begin{ex}
\emph{Let $\boldsymbol{\Gamma}_{1}$ and $\boldsymbol{\Gamma}_{2}$ be two path transformations on $\boldsymbol{\Lambda}(\M \,, \T)$ as in Definition \ref{transformationsgeneral}. If it holds that
\begin{align}
\left(  \boldsymbol{\Gamma}_{1} \,\, \circ \,\, \boldsymbol{\Gamma}_{2} \right) = \left( \boldsymbol{\Gamma}_{2} \,\, \circ \,\, \boldsymbol{\Gamma}_{1} \right), 
\end{align}
then $\boldsymbol{\Gamma}_{1}$ and $\boldsymbol{\Gamma}_{2}$ are commutative on $\boldsymbol{\Lambda}(\M \,, \T)$.
Note that if we modify the \emph{vertical bump} in (\ref{dupirevertical}) slightly such that
\begin{equation}
\label{dupireverticalmodified}
    \boldsymbol{Y}_{\T_{u,t}}(s) =
\begin{cases}
      X_s, & \text{if}\ s\in [u,t] \setminus \{\tau^*\} \\
      X_s + \epsilon, & \text{if}\ s = \tau^*,
\end{cases}
\end{equation}
where $u < \tau^* < t$, then by associating $\boldsymbol{\hat{\Gamma}}_{\epsilon}$ with (\ref{dupireverticalmodified}), instead of (\ref{dupirevertical}), we get the following:
\begin{align}
\left(  \boldsymbol{\Gamma}_{\alpha}^{(u,t)\mapsto(u,t+\alpha)}  \circ \boldsymbol{\hat{\Gamma}}_{\epsilon}^{(u,t)\mapsto(u,t)} \right) \left( \boldsymbol{X}_{\T_{[u,t]}}\right) = \left( \boldsymbol{\hat{\Gamma}}_{\epsilon}^{(u,t+\alpha)\mapsto(u,t+\alpha)}  \circ  \boldsymbol{\Gamma}_{\alpha}^{(u,t)\mapsto(u,t+\alpha)} \right) \left( \boldsymbol{X}_{\T_{[u,t]}}\right) 
\end{align}
for a fixed $\boldsymbol{X}_{\T_{u,t}}\in \boldsymbol{\Lambda}(\R , \T)$, which brings forth
\begin{align}
\left(  \boldsymbol{\Gamma}_{\alpha} \circ \boldsymbol{\hat{\Gamma}}_{\epsilon} \right) = \left( \boldsymbol{\hat{\Gamma}}_{\epsilon}  \circ  \boldsymbol{\Gamma}_{\alpha} \right) 
\end{align}
as the commutativity relation.}
\end{ex}

\begin{ex}
\emph{If $\boldsymbol{\Lambda}(\M \,, \T)$ is completely metrizable, we can define continuity of an endomorphism on $\boldsymbol{\Lambda}(\M \,, \T)$. Accordingly, let
\begin{align}
\mathcal{D}^{'}: \, \boldsymbol{\Lambda}(\M \,, \T) \times \boldsymbol{\Lambda}(\M \,, \T) \rightarrow \R \hspace{0.1in} \text{and} \hspace{0.1in}
\mathcal{D}^{''}: \, \boldsymbol{\Lambda}(\M \,, \T) \times \boldsymbol{\Lambda}(\M \,, \T) \rightarrow \R
\end{align}
be distance-metrics satisfying the requirements. Here, $\mathcal{D}^{'}$ and $\mathcal{D}^{''}$ can be the same. A path transformation $\boldsymbol{\Gamma}^{(p,r)\mapsto(s,t)}$ is continuous at $\boldsymbol{X}_{\T_{[p,r]}}\in\boldsymbol{\Lambda}(\M \,, \T)$ -- with respect to $\mathcal{D}^{'}$ and $\mathcal{D}^{''}$ -- if for every $\epsilon > 0$, there exists a $\delta > 0$ such that
\begin{align}
\mathcal{D}^{'}\left(\, \boldsymbol{X}_{\T_{[p,r]}}, \boldsymbol{X}^{*}_{\T_{[p,r]}}\, \right) < \delta \, \Rightarrow \, \mathcal{D}^{''}\left(\, \boldsymbol{\Gamma}^{(p,r)\mapsto(s,t)}\left(\boldsymbol{X}_{\T_{[p,r]}}\right), \boldsymbol{\Gamma}^{(p,r)\mapsto(s,t)}\left(\boldsymbol{X}^{*}_{\T_{[p,r]}}\,\right) \right) <\epsilon 
\end{align}
for all $\boldsymbol{X}^{*}_{\T_{[p,r]}}\in\boldsymbol{\Lambda}(\M \,, \T)$. We can further strengthen continuity; for instance, a path transformation $\boldsymbol{\Gamma}^{(p,r)\mapsto(s,t)}$ is H\"older-continuous with exponent $\alpha\in\R$, if there exists a constant $K$ such that
\begin{align}
\mathcal{D}^{''}\left(\, \boldsymbol{\Gamma}^{(p,r)\mapsto(s,t)}\left(\boldsymbol{X}_{\T_{[p,r]}}\right),\, \boldsymbol{\Gamma}^{(p,r)\mapsto(s,t)}\left(\boldsymbol{X}^{*}_{\T_{[p,r]}}\,\right) \right) \leq K \left(\mathcal{D}^{'}\left(\, \boldsymbol{X}_{\T_{[p,r]}}, \, \boldsymbol{X}^{*}_{\T_{[p,r]}}\, \right)\right)^{\alpha} 
\end{align}
for all $\boldsymbol{X}_{\T_{[p,r]}},\boldsymbol{X}^{*}_{\T_{[p,r]}}\in\boldsymbol{\Lambda}(\M \,, \T)$. When $\alpha=1$, this is Lipschitz-continuity of a path transformation.}
\end{ex}
\begin{ex}
\emph{
Let $\boldsymbol{X}_{\T_{[u,t]}}$ be an $\R$-valued \cadlag path with a jump of $\Delta$ at $\tau\in(u,t)$. Then, 
\begin{equation}
    \boldsymbol{\Gamma}^{(u,t)\mapsto(u,t)}\left(\boldsymbol{X}_{[u,t]}\right) =
\begin{cases}
			X_s, & \text{if}\ s\in [u,\tau) \\
      X_{s} - \Delta, & \text{if}\ s\in[\tau,t],
\end{cases}
\end{equation}
defines a discontinuity-removing path transformation.}
\end{ex}
\begin{ex}
\emph{The family of monotonic functionals studied in \cite{Menguturk 2025} are path transformations on unions of Skorokhod spaces of \cadlag paths. More specifically, they are projection path transformations from $\R$-valued paths to $\R$.}
\end{ex}


\end{document}